\documentclass[11pt]{article}
\usepackage[margin=1in]{geometry}
\usepackage{amsmath,amssymb,amsthm,graphicx,booktabs,hyperref,microtype}
\usepackage[T1]{fontenc}
\hypersetup{hidelinks}
\newtheorem{proposition}{Proposition}
\newcommand{\TV}{\operatorname{rowTV}}
\newcommand{\Dm}{D_{\mathrm{marg}}}
\newcommand{\Rhat}{\widehat R}
\title{Auditing Bayesian Graph Alignment:
Diagnostic Comparisons and Reference Failure}

\author{
Melika Gorgi\\
Center for Complex Biological Systems\\
University of California, Irvine\\
Irvine, CA 92697, USA\\
\texttt{mgorgi@uci.edu}
\and
Kourosh Mirsohi\\
Department of Computer Science\\
University of California, Irvine\\
Irvine, CA 92697, USA\\
\texttt{mirsohi@gmail.com}
}

\date{}

\begin{document}

\maketitle

\begin{abstract}
Bayesian graph alignment estimates correspondence probabilities, but convergence
of an alignment-score trace need not imply accurate correspondence marginals.
We audit this gap on 240 new exact graph pairs from four source families,
240 larger pairs with 20--100 vertices, and a separate 60-case
exact implementation check. Under an explicit edge-flip likelihood, we compare three samplers
and score, marginal, indicator, categorical, and classifier-based diagnostics.
Marginal disagreement improves error discrimination over score $\widehat R$
for the exact informed sampler, but its improvement for vanilla local sampling
is uncertain. Assignment-based $R^*$ and short indicator panels are competitive;
no diagnostic dominates across samplers and endpoints. At larger sizes,
diagnostics predict subsequent marginal changes, not posterior error, and
classification performance depends on the drift threshold. Disjoint-window
and held-out-chain checks attenuate but preserve positive associations.
Only 22 of 240
original reference sets pass an agreement screen. On forty failure-selected
cases, eightfold sequential Monte Carlo (SMC) particle escalation does not
resolve disagreement, whereas additional rejuvenation helps. Longer informed runs remain unstable. An
elementary feasible-alignment bound demonstrates severely unrepresentative
SMC and informed-chain scores in concentrated 100-vertex cases, independently
of approximate reference consensus. We also exhibit common-start chains with
near-zero disagreement despite exact marginal error near 0.967. These results
support assignment-sensitive auditing while identifying limits of finite
budgets, diagnostic rankings, and reference agreement as evidence of accuracy.
\end{abstract}

\section{Introduction}
Graph alignment seeks a bijection between vertices of two observed networks.
In Bayesian alignment, the output of interest may be the probability that a
particular source vertex corresponds to a particular target vertex. This
differs from finding a single good permutation: several high-scoring solutions
may distribute their mass very differently across individual assignments.
Reliable uncertainty therefore requires auditing the quantities that will be
used downstream. This distinction is consequential rather than cosmetic: a
scalar score can appear stable even when the posterior correspondence matrix---
the object used to rank alternative matches, transfer information, or propagate
alignment uncertainty---is still badly wrong. An algorithm can therefore look
converged under a familiar one-dimensional diagnostic while its reported
uncertainty remains unreliable.

Correspondence uncertainty matters when alignments support scientific decisions.
Protein-network alignment transfers evidence about functional relationships
across species~\cite{singh,berg2006,kolar2012}, while connectome matching
identifies comparable neural structures~\cite{faq}. In either setting, a high objective value alone
does not determine whether a particular node correspondence is reliable.
A posterior can express ambiguity across competing matches, but only if its
numerical approximation explores those alternatives. This motivates monitoring
assignment probabilities alongside scalar scores.

Exact enumeration provides a direct check of these probabilities on small
graphs, but its factorial cost prevents routine use at larger sizes. A useful
audit must therefore separate exact small-instance accuracy, observable
large-instance instability, and the reliability of approximate references.
Two features make this more than a sampler comparison. Where enumeration is
possible, diagnostics can be evaluated against exact posterior error rather
than an approximate proxy; where it is not, the approximate references can
be audited rather than silently treated as ground truth, and reference-free
falsification checks can still expose severe failure. This matters because a
long run or agreement between two methods can otherwise become an untested
surrogate for posterior accuracy.

We make three empirical contributions, with distinct evaluation endpoints. First, we calibrate assignment-sensitive summaries on 240 new
exact cases drawn from homogeneous, block-structured, and degree-heterogeneous
source families. Second, we assess whether early diagnostics predict later
marginal changes at 20, 50, and 100 vertices. Third, we include an informed local
kernel with a correct Hastings adjustment and explicit accounting for its
additional local evaluations. Independent long-run references are themselves
audited, and common-start experiments expose the failure of cross-chain
agreement as a certificate.
Follow-up particle/work and parallel-tempering ladder ablations on forty selected failures,
unmatched informed runs, categorical and classifier comparisons on all main
graphs, and a complete threshold sweep address the main design limitations.
An additional post hoc feasible-alignment bound checks whether observed scores
are plausible under the target without declaring an approximate reference
correct. Together, these experiments target a practical gap between
optimization-oriented graph matching and uncertainty-aware Bayesian use: they
test whether the posterior summaries actually reported to downstream analyses,
rather than only the sampled scores, are numerically credible.

The resulting benchmark characterizes when assignment-sensitive diagnostics
help and when even agreement between computational methods is misleading.
It does not introduce a new graph-matching objective or assume that a single
diagnostic can certify convergence. Code and saved experimental outputs are
provided in the companion repository.\footnote{\url{https://github.com/Mirsohi/Graph-Alignment}}
The experiments below audit a specified posterior rather than prediction
accuracy alone.

\section{Related work}
\paragraph{Graph matching and statistical recovery.}
Graph matching is commonly formulated as a quadratic assignment problem (QAP).
FAQ optimizes an indefinite continuous relaxation of this objective~\cite{faq},
while seeded matching incorporates known correspondences~\cite{fishkind}.
The behavior of convex and indefinite relaxations can differ sharply even when
the underlying discrete optimum is identifiable~\cite{lyzinski2016}.
Probabilistic soft-matching formulations have also been developed for graphs and
hypergraphs~\cite{zass2008}. Fused Gromov--Wasserstein methods instead compare
graph structure and attributes through transport couplings~\cite{fgw}; such
couplings are useful soft correspondences, but are not automatically posterior
probabilities under the likelihood studied here.

A complementary line of work studies statistical recovery under random graph
models. Recovery thresholds have been established for correlated random graphs
and related regimes~\cite{cullina,wu2022,dingdu2023}, and recent work treats
multi-graph alignment explicitly as a Bayesian estimation problem~\cite{vassaux2026}.
These results concern information-theoretic or algorithmic recoverability of a
latent alignment. Our question is different: whether a finite computation
faithfully represents the uncertainty of a specified posterior.

\paragraph{Bayesian and probabilistic graph alignment.}
Bayesian graph matching has a longer history than the recent statistical
recovery literature. Wilson and Hancock developed Bayesian compatibility and
relaxation frameworks for inexact relational graph matching~\cite{wilson1996,wilson1997},
and Williams, Wilson, and Hancock extended this perspective to multiple graph
matching~\cite{williams1997}. In biological network comparison, Berg and
L\"assig developed an evolutionarily grounded Bayesian alignment model~\cite{berg2006},
which was later implemented and extended in the \textsc{GraphAlignment}
framework~\cite{kolar2012}.

More recently, L\'azaro, Guimer\`a, and Sales-Pardo proposed a probabilistic
multiple-network alignment model that samples an ensemble of alignments and
uses posterior node-mapping probabilities rather than only a single optimum~\cite{lazaro2025}.
Gaffi, Josephs, and Lin~\cite{gaffi} develop exchangeable permutation priors,
a correlated stochastic block model, blocked Gibbs inference, and
permutation-domain posterior summaries. These models differ from one another
and from our fixed edge-flip target. Our contribution is therefore not a new
Bayesian graph-matching model, but an audit of finite-sample computational
reliability for a deliberately simple posterior whose small instances can be
checked exactly.

\paragraph{Convergence diagnostics and discrete sampling.}
For Markov chain Monte Carlo (MCMC), modern rank-normalized and folded
$\Rhat$ improves diagnosis of heavy tails and between-chain scale
differences~\cite{vehtari}. It remains a diagnostic of the monitored variable:
applying it only to a many-to-one score projection cannot restore discarded
assignment information. Multivariate convergence monitoring predates this
work~\cite{brooks}. Diagnostics specifically designed for categorical variables
compare frequency distributions with dependence adjustments~\cite{deonovic},
while classifier-based $R^*$ attempts to distinguish chains using their joint
sampled states~\cite{lambert}. Simulation-based calibration checks Bayesian
algorithms across draws from a generative model~\cite{sbc}; our exact oracle
instead measures conditional marginal error for each small graph pair. The
large-tier drift endpoint supplies neither type of posterior calibration.

Informed discrete proposals exploit neighboring target probabilities; locally
balanced functions provide a principled construction~\cite{zanella}. We use a
random-pivot restriction of a Barker-balanced proposal with an exact Metropolis
correction. The kernel is an additional audit baseline, not a claim of a new
state-of-the-art graph matcher. Sequential Monte Carlo (SMC) is also used in
graph matching~\cite{jun}; our population-annealing reference follows the same
overlap target as the MCMC methods. Gradient-informed discrete MCMC likewise
adapts proposals to local target structure~\cite{grathwohl}; we do not benchmark
that broader class. Our reference designs draw on annealed SMC samplers~\cite{smc}
and parallel tempering~\cite{earl}, whose use of multiple particles or
temperatures does not by itself establish adequate exploration.

\section{Targets and graph families}
For undirected binary adjacency matrices $A,B$ and $\pi\in S_n$, define
\[
 S(\pi)=\sum_{i<j}A_{ij}B_{\pi(i)\pi(j)},\qquad
 \mu_\beta(\pi\mid A,B)=Z_\beta^{-1}\exp\{\beta S(\pi)\}.
\]
The correspondence marginal is $P_{ia}=\Pr_{\mu_\beta}(\pi(i)=a)$, and
$\TV(P,Q)=\|P-Q\|_1/(2n)$, with the norm summing all matrix entries.
The planted permutation evaluates correspondence to the generated identity;
it is not substituted for posterior truth.

\paragraph{A likelihood valid for structured sources.}
Condition on a source graph $A$, draw the permutation uniformly, and flip each
edge or nonedge independently with probability $q<1/2$ in the permuted copy.
The mismatch count satisfies
\[
 M(\pi)=|E(A)|+|E(B)|-2S(\pi).
\]
Consequently $\log p(B\mid A,\pi)=C+\beta S(\pi)$ with
$\beta=2\log((1-q)/q)$. This conditional likelihood is valid for every
distribution of $A$. It avoids applying a homogeneous resampling likelihood
to a different structured observation process.

The new exact tier uses an Erd\H{o}s--R\'enyi (ER) model with edge probability
$0.3$, sparse ER with $\min(0.3,1.5/n)$, a two-block stochastic block model
(SBM) with block proportions approximately $0.4/0.6$ and within/between
probabilities $0.55/0.08$, and a degree-heterogeneous SBM.
In the latter, independent vertex weights are uniform on $[0.4,1.6]$, multiply
the block probabilities, and probabilities are capped at $0.95$. Block labels
and degree weights are not given to the samplers. Unequal block sizes reduce
one source of trivial block-exchange symmetry, but small realized graphs can
still have automorphisms. This tier is not restricted to rigid graphs.

A separate implementation check uses 60 ER pairs at $n\in\{8,9,10\}$,
twenty seeds per size, with edge probability $p=0.3$. Each source edge or
nonedge is independently redrawn from $\operatorname{Bernoulli}(p)$ with
probability $\epsilon=0.1$ before relabeling. Its overlap posterior has
\[
 \beta=\log\frac{(1-\epsilon+\epsilon p)(1-\epsilon p)}
 {\epsilon^2p(1-p)}.
\]
This resampling channel differs from independent edge flips; its results are
reported separately. Four local Metropolis--Hastings (MH) chains per pair use $2^{18}$ proposals,
discard $2^{16}$, and retain every fifth subsequent state. Saved exact-oracle
errors are reused for this implementation check; the main 240-case tier
independently recomputes its exact posterior.

\section{Assignment diagnostics and their interpretation}
For $K$ independent chains with empirical marginal matrices $\widehat P^{(a)}$,
define
\[
 \Dm=\frac{2}{K(K-1)}\sum_{a<b}\TV(\widehat P^{(a)},\widehat P^{(b)}).
\]
Constructing the matrices costs $O(KTn)$ for $T$ retained permutations per
chain; their dense storage costs $O(Kn^2)$ and direct pairwise comparison
costs $O(K^2n^2)$. No exact posterior or planted correspondence is required.

We compare $\Dm$ with rank-normalized/folded score $\Rhat$, the reciprocal of the
score bulk effective sample size (ESS), $1/\mathrm{ESS}_{\mathrm{bulk}}$,
maximum indicator $\Rhat$ in a fixed panel of at most sixteen assignments,
normalized Frobenius disagreement, and mean row Jensen--Shannon divergence.
The indicator panel selects $(i,(7i+3)\bmod n)$ for
$i<\min(n,16)$ before observing samples. Its limited coverage is reported;
it is not an exhaustive assignment-wise diagnostic. Constant traces yield an
undefined diagnostic, not automatic convergence. Analyses record missingness
and additionally rank undefined score diagnostics as warnings in a sensitivity
analysis.

\begin{proposition}[What disagreement lower bounds]
For any marginal matrix $P$,
$\Dm/2\le K^{-1}\sum_a\TV(\widehat P^{(a)},P)$.
This need not lower bound $\TV(K^{-1}\sum_a\widehat P^{(a)},P)$.
\end{proposition}
\begin{proof}
Apply the triangle inequality to each pair. Every individual chain error
appears in $K-1$ pairs. Opposite chain biases can cancel in the pooled matrix,
so the analogous pooled-error conclusion does not follow.
\end{proof}

\begin{proposition}[Independent-draw calibration]
If each chain has $T$ independent and identically distributed (IID) posterior draws, then
\[
 \mathbb E\Dm\le\min\{1,\sqrt{(n-1)/(2T)}\},\qquad
 \Pr(\Dm-\mathbb E\Dm\ge t)\le e^{-KTt^2/2}.
\]
\end{proposition}
\begin{proof}
An entrywise difference has variance $2P_{ia}(1-P_{ia})/T$.
Jensen followed by rowwise Cauchy--Schwarz gives the expectation bound.
Replacing one sampled permutation changes its empirical matrix by row-TV
at most $1/T$, hence changes $\Dm$ by at most $2/(KT)$.
McDiarmid's bounded-differences inequality gives the tail bound.
\end{proof}
These are IID statements. Replacing $T$ by the score ESS is not justified.
For independent stationary, dependent chains, the expectation is instead
bounded by the pairwise average of
$(2n)^{-1}\sum_{ia}\sqrt{V_{a,ia}+V_{b,ia}}$, where each $V$ is the
variance of its chain's marginal estimate. The relevant autocorrelations are
assignment-specific.

\paragraph{Stationary variance identity.}
For a stationary indicator trace of length $T$ with variance
$v=P_{ia}(1-P_{ia})$ and autocorrelations $\rho(h)$,
\[
 \operatorname{Var}(\widehat P_{ia})=\frac vT
 \left[1+2\sum_{h=1}^{T-1}(1-h/T)\rho(h)\right].
\]
This is the finite-length covariance-sum identity. For independent chains,
variances of entrywise differences add; Jensen's inequality then gives the
dependent-draw expectation bound above. We do not claim a dependent-draw
concentration bound without additional assumptions.

\paragraph{Finite-sample error is not automatically bias.}
For $N$ IID posterior permutations the exact expected marginal row-TV is
\[
 \frac1n\sum_{ia}P_{ia}(1-P_{ia})
 \Pr\{\operatorname{Bin}(N-1,P_{ia})=\lfloor NP_{ia}\rfloor\}.
\]
For completeness, let $X\sim\operatorname{Bin}(N,p)$ and
$k=\lfloor Np\rfloor$. Since $\mathbb E(X-Np)=0$,
$\mathbb E|X-Np|=2\mathbb E[(X-Np)_+]$. Using
$j\Pr(X=j)=Np\Pr\{\operatorname{Bin}(N-1,p)=j-1\}$ in the positive tail
and the binomial recursion yields
\[
 \mathbb E|X/N-p|=2p(1-p)
 \Pr\{\operatorname{Bin}(N-1,p)=k\}.
\]
Each marginal count from IID permutations has this binomial law. Summing the
entrywise expectations and dividing by $2n$ proves the displayed floor; no
independence across entries is required for this expectation. The implementation
verifies the identity against direct binomial summation at several sample sizes
and probabilities. We use this deterministic mean floor for all new exact cases
and distinguish it from a quantile or uncertainty interval.

\section{Samplers and experimental protocol}
Random-transposition MH proposes uniformly and accepts with
$\min(1,e^{\beta\Delta S})$. Parallel tempering (PT) uses eight replicas on the quadratic ladder
$\beta_r=\beta(r/7)^2$ with adjacent exchanges. All main-study initial states
are independent uniform permutations within each set of four runs.

\paragraph{Informed local baseline.}
Choose pivot $u$ uniformly and weight its $n-1$ swaps by
$g(e^{\beta\Delta S})$, where $g(t)=t/(1+t)$. Conditional on the pivot,
the proposal is normalized by $Z_u(\pi)$. Since $g(t)=t g(1/t)$, the
Hastings acceptance is $\min\{1,Z_u(\pi)/Z_u(\pi')\}$. Each conditional
kernel is reversible, so their uniform mixture is reversible. The implementation
uses $2(n-1)+1$ score-delta evaluations per iteration, including both neighbor
sets and the selected delta. We count all of these evaluations. The stronger
proposal thus receives fewer iterations under matched work and need not yield
smaller error.

\paragraph{Exact tier.}
At $n=8$ we use four source families, three flip probabilities
$q\in\{0.10,0.30,0.45\}$, and twenty graph seeds per condition: 240 graphs.
Each candidate method receives four chains, each with at most $2^{18}$ delta
evaluations. Exact marginals, entropy, maximum a posteriori (MAP) mass, and MAP multiplicity are
computed once per graph. The exact oracle enumerates the same overlap target
as all candidate methods.

\paragraph{Large tier.}
We use $n\in\{20,50,100\}$, ER and SBM, $q\in\{0.10,0.30\}$, and twenty
seeds per condition: another 240 graphs. Early estimates use $2^{16}$ delta
evaluations per chain. Later estimates extend the identical seeded trajectory
to sixteen times the early work. Both estimates discard the first quarter
of their respective runs, so the later retained window starts after the early
window ends. Approximately 1024 permutations are retained per chain; thinning
bounds storage rather than being claimed to improve efficiency.

Independent references use separate seeds for four long PT runs and four
population-annealing SMC runs. SMC has 512 particles and 64 annealing stages;
rejuvenation work is matched to the long-run delta budget. We record between-method and within-method marginal disagreement. An operational reference gate
requires all three discrepancies to be at most $0.05$. This gate tests empirical
agreement and resolution, not convergence: finite particle count alone can
prevent it from passing on diffuse large targets. Neither passing nor failing
isolates all causes of reference error.

\paragraph{Endpoints and statistics.}
The primary exact endpoint is pooled marginal error; the large endpoint is
distance between early and later pooled marginals. Later drift measures
instability, not distance to the posterior. We report Spearman association, the area under the receiver operating
characteristic curve (ROC AUC) at an explicit endpoint threshold $0.1$, average precision, and
prevalence. Bootstrap intervals use 1000 resamples of graph instances within
each family/size/noise cell. Sampler comparisons are paired by graph, and
diagnostic AUC differences use the same valid cases. Per-condition summaries
expose pooling effects; no population-wide causal interpretation is assigned
to a concentration association. AUC is undefined when there are no positive
or no negative cases.
When every case is positive, average precision equals its prevalence baseline
of one and is not evidence of discrimination. In the large tier there are only
five PT cases below the drift threshold; its pooled AUC contrast is therefore
based on few negative cases and must be read alongside condition-specific
rank correlations.

\subsection{Follow-up design: reference budgets, diagnostics, and thresholds}
We fix a failure-selected subsample before follow-up sampling: five graphs
drawn uniformly without replacement from each size ($50,100$), source family
(ER, SBM), and noise ($0.1,0.3$) cell, for forty graphs. Selection conditions
on failure of the original reference screen. Thus follow-up pass rates describe
these forty cases and do not estimate the original 240-case population rate.

\paragraph{Reference work and resolution.}
For each selected graph we run four independent SMC replicates at
$512,1024,2048,4096$ particles, retaining 64 annealing stages and 32 local
rejuvenation moves per particle per stage. The 512-particle baseline reproduces
the saved original marginals. Particle escalation also escalates work. Two
controls distinguish that tradeoff: 4096 particles with four moves use baseline
work, and 512 particles with 256 moves use the same work as the largest
32-move run. We save weight ESS values, surviving initial ancestors, distinct final
particles, scores, and marginal discrepancies. Ancestry loss diagnoses particle
genealogy, not by itself bias after rejuvenation.

We also rerun PT with 8, 16, and 32 quadratic-ladder temperature replicas.
Each independent PT run makes $2^{22}$ local proposals \emph{in total across
all temperature replicas}; it does not give $2^{22}$ proposals to each
temperature replica. One sweep proposes one move at each of the $R$ replicas,
so the respective sweep counts are $524288$, $262144$, and $131072$.
There are four independently seeded PT runs per graph and ladder setting,
for $2^{24}$ local proposals across those four runs. Relative to the original
8-replica reference, this quadruples local-proposal work; across new ladders it
holds that work fixed. More replicas therefore mean fewer sweeps per replica.
Per-edge swap rates and cold--hot--cold walker trips distinguish local exchange
acceptance from traversal of the entire ladder. None of these interventions
alone identifies posterior multimodality or proves that either reference is
correct. Our screen still requires both within-method and cross-method row-TV
discrepancies at most $0.05$.

\paragraph{Unmatched informed runs.}
On the same forty graphs, four informed chains receive 10,000, 40,000, and
160,000 transitions, with the same initializations and nested random streams
across budgets. Each run discards its first quarter and retains up to 4096
draws per chain. We report discrepancy from the preceding budget, disagreement
across chains, and drift between halves of the retained window. Runs stop at
the prespecified maximum; apparent within-basin stability is not a stopping
certificate. These budgets address starvation without asserting unlimited
computation or an optimally tuned informed sampler.

\paragraph{Categorical and classifier comparators.}
For all 480 original graphs and all three methods, we regenerate early samples
with the original seeds and verify their marginal matrices against the saved
ones before computing new diagnostics. We implement the Wei\ss/DAR(1)
dependence adjustment described by Deonovic and Smith~\cite{deonovic} for each
node's categorical assignment. If $p_{cj}$ is chain $c$'s frequency of category
$j$, $\bar p_j$ its pooled frequency, $K$ the chain count, and $T$ the number
of draws per chain, then
\[
 X^2=T\sum_{c,j:\bar p_j>0}\frac{(p_{cj}-\bar p_j)^2}{\bar p_j},\qquad
 X^2_{\mathrm{adj}}=X^2\frac{1-\hat\phi}{1+\hat\phi}.
\]
Here $\hat\phi=1+(KT)^{-1}-(1-\hat s)/(1-\sum_j\bar p_j^2)$ and $\hat s$
is the pooled probability of identical consecutive assignments. Negative
estimates are truncated to zero, consistent with the DAR(1) parameter domain.
Zero degrees of freedom or $\hat\phi\ge1$ are flagged as undefined. We rank
graphs by the largest valid adjusted statistic divided by its degrees of
freedom $(J-1)(K-1)$, and retain Bonferroni-adjusted nodewise $p$-values,
undefined-coordinate fractions, and sparse expected-count fractions. Actual
permutation chains need not follow DAR(1); these are model-based diagnostics,
not calibrated significance guarantees for this application.

For $R^*$~\cite{lambert}, each chain is split into two halves, producing eight
class labels. A stratified random 70/30 split trains and evaluates a random
forest with 100 trees, minimum leaf size five, square-root feature subsampling,
and a fixed graph-specific seed. We report eight times test accuracy using
three inputs: the score, assignments of the first sixteen nodes (all nodes
when $n<16$), and the full assignment. Assignment categories are one-hot
encoded, so numerical vertex labels do not impose an artificial order. This
node panel uses categorical assignments and is richer than the original
sixteen binary assignment indicators; it is not a feature-count-matched test.
This implements a fixed random-forest (RF) version of the classification diagnostic, not an
exhaustive classifier search or its uncertainty ensemble. Random train/test
splits of autocorrelated draws are not independent replication; $R^*$ is used
as a discrimination summary, not a significance test.

\paragraph{Threshold robustness.}
We retain all cutoffs in $\{0.05,0.1,0.2,0.3,0.4,0.5,0.6,0.7,0.8,0.9\}$, reporting
positive/negative counts, prevalence, AUC, and average precision for the three
original diagnostic panels, pooled and separately by size. AUC intervals and
new paired diagnostic contrasts use 500 graph bootstrap resamples within
size/family/noise cells; SMC paired contrasts use 2000. The main-study intervals
retain their original 1000 resamples. Spearman association with continuous
drift is independent of the cutoff. We do not choose a cutoff after observing
performance or substitute an unknown large-graph posterior's IID floor.

\paragraph{Multiplicity and analysis status.}
The follow-up program includes several samplers, budgets, diagnostics,
thresholds, and overlapping graph subsets. Its intervals are pointwise,
exploratory summaries; they do not control familywise error across this full
program. The fixed follow-up and revision protocols were recorded before
their respective new runs, but followed inspection of earlier results and
are not independent confirmatory studies. Reusing graphs enables paired
contrasts and also makes these comparisons dependent. We report the complete
specified grids rather than interpreting every interval that excludes zero
as a separate confirmed discovery.

\subsection{Disjoint samples and broader moves}
\paragraph{Nonoverlapping early windows.}
For all 240 larger graphs and three samplers, regenerate and verify the
original retained early trajectories. From each chain with $T$ retained draws,
let $h=\lfloor T/2\rfloor$. Window A consists of the first $h$ draws and
window B of the last $h$; any middle draw is dropped. Compute diagnostics from
A and define the outcome as $\TV(\bar P_B,\bar P_{\mathrm{later}})$, using the
original later window. This removes literal reuse of draws between diagnostic
and outcome. Adjacent windows and the subsequent trajectory can remain
dependent; this is not an independent-chain validation.
For a matched-draw-count control, compare the same A diagnostic against
$\TV(\bar P_A,\bar P_{\mathrm{later}})$. Also reverse A and B for marginal
disagreement. All seven diagnostic panels are recomputed from A, with their
original hyperparameters. Pointwise correlation intervals and paired changes
use 1000 graph bootstrap resamples within size/family/noise cells. Separate
size-only, family/noise-only, and full-condition decompositions of the original
categorical result assess aggregation effects without assigning a unique cause.
An analysis-only secondary check, added before inspecting the split-window
outcomes, uses disagreement of original chains 0/1 against drift of chains 2/3,
and reverses those groups. The corresponding shared control uses the same
two chains for both quantities. This separates random streams as well as
draws, at the cost of only two chains per estimate; all comparisons still
condition on the same graph and experimental design.

\paragraph{A broader move-class probe.}
Select two of the previous $n=100$ failure cases from each family/noise cell
before new sampling (eight graphs). At each transition, independently of the
state, generate a candidate multiset of eight random transpositions and eight
products of two disjoint transpositions. Conditional on this multiset $\mathcal M$,
every candidate $m$ is an involution. Choose $m$ with probability proportional to
$g\{\mu(m x)/\mu(x)\}$, where $g(t)=t/(1+t)$, and accept with
$\min\{1,Z_{\mathcal M}(x)/Z_{\mathcal M}(m x)\}$, evaluating both normalizers
on the \emph{same} candidate multiset. For fixed $\mathcal M$, the Barker ratio
identity yields detailed balance, including duplicate candidates. A
state-independent mixture of these kernels remains reversible; single swaps
allow both permutation parities. An exhaustive five-vertex transition-matrix
check verifies the implemented conditional detailed balance and composed deltas.

This changes the move class and candidate construction, so it is an exploratory
alternative rather than a one-factor isolation of pivot choice. Four chains
reuse the original initializations and receive either 160,000 transitions or
663,333 transitions. Each transition evaluates 48 transposition deltas: eight
single and eight double moves in each of two normalizers. The latter budget uses 31,839,984 deltas per chain, compared with 31,840,000 for the original
single-pivot kernel at 160,000 transitions; the former matches transition
counts. Quarter burn-in and at most 4096 retained draws match the earlier
protocol. Neither comparison equates all computational overhead or optimizes
the broader family of informed proposals.

\paragraph{Implementation and reproducibility checks.}
An exhaustive five-vertex check compares all score deltas against recomputed
scores and compares the declared temperature target with the direct flip
likelihood. A four-vertex transition matrix for the informed kernel verifies
detailed balance and stationarity to numerical precision, while the broader
move-class implementation is checked exhaustively at five vertices for
conditional detailed balance and composed deltas. Sampled score traces are
checked against direct scores for each kernel. Diagnostic checks include
independent traces, shifted chains, and constant traces, and structural and
numerical consistency checks of the saved implementation-check outputs also
pass. The companion repository contains per-condition tables, diagnostic
missingness checks, paired bootstrap summaries, exact IID floors, and
reference-gate outcomes.

\section{Results}
\paragraph{Exact implementation check.} All sixty resampling-channel local MH cases were rerun. Marginal disagreement matches the saved values exactly. Corrected rank normalization with one constant trace treated as undefined (59 valid score cases) gives score diagnostic correlation 0.483 and AUC 0.749, versus 0.926 and 0.969 for $\Dm$. The corrected and previous rank-normalization conventions give the same substantive ranking. Exact error labels in this replication come from the previously validated oracle outputs; the new 240-case tier recomputes its own oracle.

\paragraph{New exact calibration.} Table~\ref{tab:exact} reports the diagnostic comparisons for 240 graphs per method. High-error cases are uncommon for vanilla local MH and absent for PT at this budget. Consequently a universal AUC improvement is not supported. For informed MH the paired AUC difference between marginal disagreement and score $\Rhat$ is 0.129 (95\% interval [0.034, 0.235]). The fixed assignment-indicator baseline is competitive; the result supports assignment sensitivity rather than a unique advantage of marginal total variation.

\begin{table}[tb]\centering\small\begin{tabular}{llrrr}\toprule Method & Diagnostic & Spearman & AUC & High error \\\midrule

Local MH & $D_{\rm marg}$ & 0.527 & 0.976 & 4/240 \\

Local MH & Score $\widehat R$ & 0.137 & 0.944 & 4/240 \\

Local MH & Indicator $\widehat R$ & 0.409 & 0.972 & 4/240 \\

PT & $D_{\rm marg}$ & 0.271 & not estimable & 0/240 \\

PT & Score $\widehat R$ & 0.013 & not estimable & 0/240 \\

PT & Indicator $\widehat R$ & 0.031 & not estimable & 0/240 \\

Informed MH & $D_{\rm marg}$ & 0.716 & 0.975 & 16/240 \\

Informed MH & Score $\widehat R$ & 0.439 & 0.846 & 16/240 \\

Informed MH & Indicator $\widehat R$ & 0.623 & 0.971 & 16/240 \\

\bottomrule\end{tabular}\caption{New exact tier; high error means pooled row-TV above 0.1. AUC is undefined without both label classes.}\label{tab:exact}\end{table}

\paragraph{Posterior concentration.} Across the 240 exact instances, posterior entropy and the local-minus-PT marginal-error gap have Spearman correlation $\rho=-0.386$. In a descriptive regression controlling for source family, flip-probability category, and realized source density, the entropy coefficient is $-0.00624$ with stratified bootstrap interval $[-0.00986,-0.00296]$. Size is fixed at eight. These associations are consistent with a regime effect but do not establish a causal role for entropy or a universal sampler ranking.

\paragraph{Larger graphs.} All 240 planned larger graph pairs completed. Table~\ref{tab:large} compares early diagnostics against subsequent marginal drift. Differences in pooled rank correlation can reflect size and noise as well as within-condition variation; the companion repository reports every condition separately. When every run exceeds the drift threshold, AUC cannot distinguish diagnostics in that method group.

\begin{table}[tb]\centering\small\begin{tabular}{lrrrrr}\toprule Method & $\rho(D_{\rm marg})$ & $\rho(\widehat R)$ & AUC$(D_{\rm marg})$ & AUC$(\widehat R)$ & Drift $>0.1$ \\\midrule

Local MH & 0.433 & 0.184 & 0.998 & 0.995 & 200/240 \\

PT & 0.809 & 0.437 & 1.000 & 0.814 & 235/240 \\

Informed MH & 0.665 & 0.210 & not estimable & not estimable & 240/240 \\

\bottomrule\end{tabular}\caption{Large-tier endpoint is drift after sixteen times the early work, not exact posterior error.}\label{tab:large}\end{table}

\paragraph{Reference resolution.} 22 of 240 graph pairs pass the operational PT/SMC agreement gate. Mean cross-method reference distances are $n=20$: 0.056, $n=50$: 0.691, and $n=100$: 0.912. These runs do not supply a generally validated posterior truth set. In particular, the SMC reference has only 512 particles per run; Monte Carlo resolution and genealogy can contribute to disagreement even when a target is diffuse. We therefore base the principal large-scale conclusion on measured drift and report reference distance only as a screened secondary comparison. The data do not certify accurate large-graph marginals.
\paragraph{Matched replication and condition-level checks.}
On the same 59 evaluable original cases, marginal disagreement has Spearman
correlation 0.928 and AUC 0.968. The difference from the previous 60-case score
summary reflects both the small plotting-position correction and the explicit
handling of the degenerate case. Within the twelve large-tier size/family/noise
conditions, disagreement has higher rank correlation with drift than score
$\widehat R$ in 11 conditions for local MH, 10 for PT, and 9 for informed MH.
Median within-condition correlations for disagreement versus score are
0.392 versus 0.035, 0.424 versus 0.138, and 0.274 versus -0.103, respectively.
These descriptive comparisons do not constitute twelve independent hypothesis
tests with multiplicity-adjusted significance guarantees.

\subsection{Reference-budget and diagnostic follow-ups}

\begin{table}[tb]
\centering\small
\begin{tabular}{lrrrrr}
\toprule
Particles & Moves & Cross PT & Within SMC & Pass & Ancestors \\
\midrule
512 & 32 & 0.788 & 0.903 & 0 & 4.556 \\
512 & 256 & 0.721 & 0.735 & 0 & 11.025 \\
1024 & 32 & 0.790 & 0.891 & 0 & 6.575 \\
2048 & 32 & 0.792 & 0.886 & 0 & 10.244 \\
4096 & 4 & 0.850 & 0.936 & 0 & 6.456 \\
4096 & 32 & 0.784 & 0.874 & 0 & 19.306 \\
\bottomrule
\end{tabular}
\caption{SMC on forty failure-selected graphs. Distances are mean row-TV; passes require all original screen conditions. Ancestors are surviving initial particle lineages, averaged over graphs and four replicates. The 4096/4 and 512/256 settings are work controls.}
\label{tab:smc-followup}
\end{table}

\begin{table}[tb]
\centering\small
\begin{tabular}{lrrrrr}
\toprule
Replicas & Cross original PT & Within PT & Trips & Min. swap & Cold swap \\
\midrule
8 & 0.508 & 0.761 & 5.013 & 0.004 & 0.200 \\
16 & 0.710 & 0.761 & 8.569 & 0.082 & 0.386 \\
32 & 0.704 & 0.740 & 9.150 & 0.322 & 0.618 \\
\bottomrule
\end{tabular}
\caption{Longer PT references on the same forty graphs. Trips are mean cold--hot--cold traversals per replicate; minimum swap is the graphwise minimum over edges and replicates, averaged over graphs. Replica count trades ladder resolution against sweeps at fixed local-proposal work.}
\label{tab:pt-followup}
\end{table}

\begin{table}[tb]
\centering\small
\begin{tabular}{lrrrr}
\toprule
Transitions & Adjacent drift & $D_{\mathrm{marg}}$ & Window drift & Acceptance \\
\midrule
10,000 & 0.701 & 0.936 & 0.309 & 0.125 \\
40,000 & 0.355 & 0.881 & 0.277 & 0.103 \\
160,000 & 0.308 & 0.821 & 0.222 & 0.093 \\
\bottomrule
\end{tabular}
\caption{Unmatched informed runs on forty selected graphs. The first adjacent comparison uses the original early budget; subsequent rows compare successive new budgets. Drift and disagreement are different quantities and neither alone certifies accuracy.}
\label{tab:informed-followup}
\end{table}

\begin{table}[tb]
\centering\small
\begin{tabular}{lrrrr}
\toprule
Diagnostic & Exact informed AUC & Large local MH $\rho$ & Large PT $\rho$ & Large informed MH $\rho$ \\
\midrule
$D_{\mathrm{marg}}$ & 0.975 & 0.433 & 0.809 & 0.665 \\
Score $\widehat R$ & 0.846 & 0.184 & 0.437 & 0.210 \\
Indicator $\widehat R$ & 0.971 & 0.302 & 0.328 & 0.190 \\
Wei\ss/DAR(1) & 0.907 & -0.355 & 0.254 & -0.745 \\
Score $R^*$ & 0.810 & 0.251 & 0.692 & 0.342 \\
Panel $R^*$ & 0.973 & 0.747 & 0.566 & -0.043 \\
Full $R^*$ & 0.973 & 0.720 & 0.599 & 0.061 \\
\bottomrule
\end{tabular}
\caption{Added diagnostic comparison. Exact AUC uses error above 0.1; large Spearman correlations use continuous later drift. These columns describe different endpoints and sampler groups. Undefined values are excluded and their counts are retained in the companion repository. Paired bootstrap contrasts use shared valid graphs.}
\label{tab:comparators-followup}
\end{table}

\begin{figure}[tb]
\centering\includegraphics[width=\linewidth]{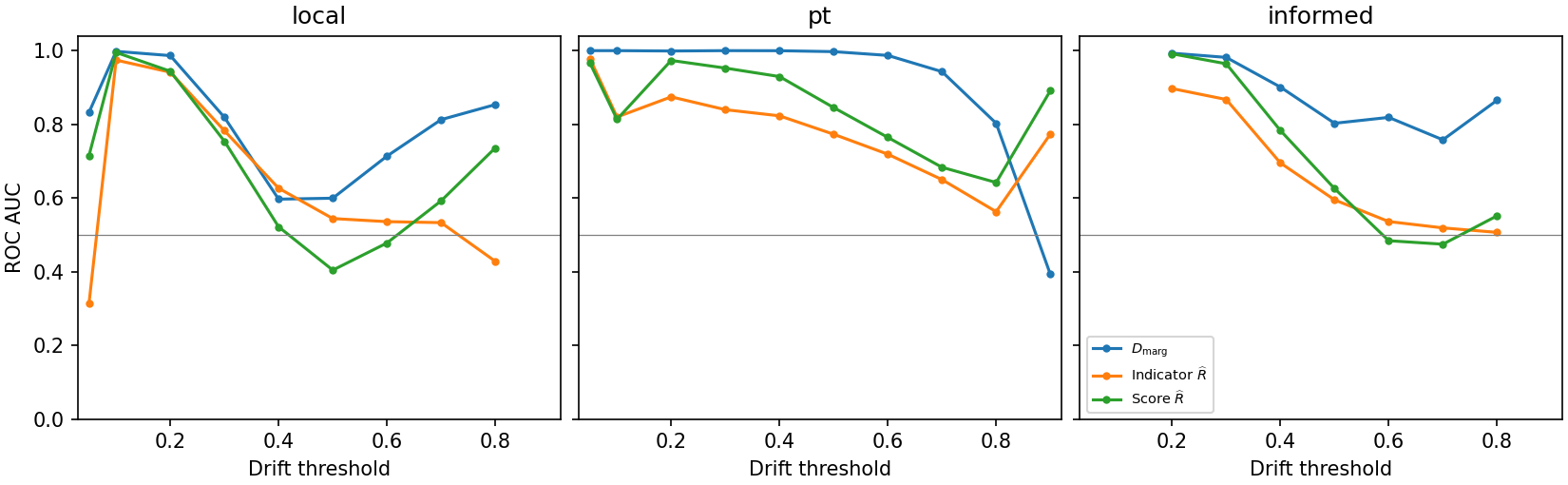}
\caption{Prespecified large-tier drift-threshold sweep for the original diagnostic panels. Missing points have undefined AUC. All positive/negative counts and stratified bootstrap intervals, including size-specific results, are available in the companion repository. Spearman correlation with continuous drift does not change with threshold.}
\label{fig:threshold-followup}
\end{figure}

\paragraph{What the reference ablation identifies.}
Increasing particles alone from 512 to 4096 leaves mean discrepancy from the
original PT reference nearly unchanged ($0.788$ versus $0.784$; paired change
$-0.004$, stratified bootstrap 95\% interval $[-0.014,0.007]$), and no setting
passes the screen on any of the forty selected cases. At the same eightfold
work as 4096 particles with 32 moves, using 512 particles with 256 moves reduces
mean cross-reference discrepancy to $0.721$ (change $-0.067$, interval
$[-0.086,-0.041]$) and within-SMC disagreement from
$0.903$ to $0.735$. Allocating baseline work to more particles and only four
moves instead raises cross-reference discrepancy to $0.850$. Thus particle
count alone is insufficient in this design, while rejuvenation allocation
materially affects agreement. This is evidence of a computational allocation
limitation, not proof that the more concordant reference is accurate.
Only about 19 initial lineages survive on average even with 4096 particles
and 32 moves; genealogy remains concentrated despite the larger population.

Increasing PT ladder resolution improves mean minimum-edge swap acceptance
from $0.004$ at eight replicas to $0.322$ at thirty-two. Mean round trips increase
from about $5.0$ to $9.2$ per replicate, but within-PT disagreement remains about
$0.74$--$0.76$. Better exchange telemetry therefore does not resolve reference
uncertainty at the tested work. SMC and PT both retain substantial independent
run disagreement, so movement toward the original PT answer cannot by itself
be interpreted as movement toward the posterior.
Comparing each new PT ladder against each particle level also yields no screen
passes. At 4096 particles, mean cross-method discrepancy remains $0.758$, $0.752$,
and $0.737$ for 8, 16, and 32 replicas, respectively.

The absolute screen also has a resolution limitation. For two independent
$N$-sample estimates from a uniform-permutation target, let
$F_{N,1/n}$ denote the binomial CDF. Their exact expected row-TV is
\[
 \frac{n}{N}\sum_{k=0}^{N-1}F_{N,1/n}(k)\{1-F_{N,1/n}(k)\}.
\]
At $n=100$, this is $0.245$ for $N=512$ and $0.088$ for $N=4096$;
even the latter exceeds the $0.05$ screen. These are illustrative expectations,
not estimates of the unknown graph-specific posterior floor. They explain
why failure of a fixed absolute gate need not identify nonconvergence. The
much larger observed discrepancies, work sensitivity, and genealogy and
ladder telemetry provide additional algorithmic evidence beyond gate counts.

\paragraph{A constructive check of reference inadequacy.}
As a post hoc analysis of the saved scores, a feasible alignment supplies a
useful bound without an approximate reference. Under the uniform permutation
prior, let $s_w$ be its score and let $s$ be any threshold. Then
\[
 \Pr_{\pi\mid A,B}\{S(\pi)\le s\}
 \le \min\{1,n!\exp[\beta(s-s_w)]\}.
\]
Indeed, the sublevel numerator is at most $n!\exp(\beta s)$ and the
normalizer is at least $\exp(\beta s_w)$. The witness need not be MAP.
This bounds the entire score sublevel, not just the finite observed support.
We use the saved planted alignment as a feasible witness and evaluate the
bound in log space at the largest final SMC particle score across all four
replicates. For all ten selected $n=100,q=0.1$ graphs, under every one of the
six SMC settings, the sublevel containing every final particle has posterior
mass below $10^{-532}$. Thus in this concentrated stratum, the final SMC
scores are demonstrably unrepresentative of the target; disagreement is not
merely a consequence of the marginal estimator's finite resolution. The bound
is uninformative in the $q=0.3$ stratum and does not validate PT or quantify
node-marginal error. It also does not identify a unique landscape mechanism.

\paragraph{The informed result survives budget escalation as a limitation.}
Increasing the informed budget reduces mean adjacent-budget drift from $0.701$
at the first escalation to $0.308$ at 160,000 transitions, while mean cross-chain
disagreement falls from $0.936$ to $0.821$. At the final budget, mean within-window
drift is still $0.222$, and none of the forty graphs has both adjacent drift and
disagreement at most $0.1$. This is progress in stability, not evidence that the
chains have mixed. The original 329-transition row at $n=100$ was starved, but
that matching constraint is not the sole explanation for the remaining failure
at the tested unmatched budgets.

The same post hoc witness bound applied to saved informed score traces further
qualifies this result. For every selected $n=100,q=0.1$ graph at 160,000
transitions, the score sublevel containing all retained draws has posterior
mass below $10^{-521}$. Longer runs of this kernel therefore remain severely
unrepresentative in that stratum. This conclusion concerns the implemented
kernel, initialization, and finite budgets, not all informed samplers or an
asymptotic impossibility of mixing.

\paragraph{Threshold robustness is conditional, not uniform.}
At cutoff $0.5$, the informed large tier has 153 positives and 87 negatives,
with disagreement AUC $0.803$, replacing the undefined AUC at $0.1$ with an
estimable comparison. Across informative cutoffs $0.2$--$0.8$, its disagreement
AUC ranges from $0.758$ to $0.993$. PT remains highly discriminative through
moderate cutoffs, but its AUC at $0.9$ is only $0.393$ and has one positive case.
Local MH disagreement AUC also falls to about $0.60$ at cutoffs $0.4$ and $0.5$.
Thus the audit is sensitive to what magnitude of instability is classified
as a failure. The full grid and continuous correlations support a conditional
interpretation; no uniformly favorable threshold is claimed.

\paragraph{Assignment information helps; no single diagnostic dominates.}
On exact informed cases, full-assignment $R^*$ has AUC $0.973$, essentially
matching disagreement at $0.975$. Their paired AUC difference is $0.001$ with
95\% interval $[-0.016,0.019]$. The node-panel and full versions coincide at
$n=8$ because both include every node. Score-only $R^*$ has AUC $0.810$, while
the categorical statistic has AUC $0.907$; disagreement's AUC advantage over
the latter has an interval spanning zero ($[-0.007,0.195]$). The original indicator panel with at most sixteen assignments is also competitive at $0.971$. Thus the exact result
supports monitoring assignment information, not a unique advantage of full
marginal disagreement over other assignment-aware summaries.

At larger sizes, the ranking depends on sampler and aggregation. For local MH,
node-panel and full $R^*$ have pooled drift correlations $0.747$ and $0.720$,
above disagreement's $0.433$. The paired difference between disagreement and
the panel is $-0.314$ with interval $[-0.362,-0.267]$. Within the twelve
size/family/noise conditions, however, median correlations are $0.392$ for
disagreement, $0.171$ for panel $R^*$, and $0.094$ for full $R^*$. Both summaries
matter; the pooled advantage does not translate into a within-condition
advantage. For PT and informed MH, disagreement's pooled correlations are
$0.809$ and $0.665$, compared with $0.599$ and $0.061$ for full $R^*$.
Score-only classification partially closes the PT gap (correlation $0.692$),
but does not eliminate it.

Poor ranking of drift magnitude should not be confused with failure to flag
nonconvergence. Full $R^*$ is at least $7.9$ out of a maximum of eight in
131 of 240 large informed cases: chains are readily identifiable, but the
classifier has little remaining range for ordering severity. The categorical
statistic's pooled informed correlation is negative ($-0.745$), whereas its
median within-condition correlation is $0.084$. Its DAR(1) adjustment, sparse
categories, unequal transition budgets across sizes, and this particular
maximum-over-nodes aggregation limit interpretation. These comparisons do
not establish that categorical diagnostics or alternative classifiers are
generally ineffective. All native categorical $p$-values and diagnostic
coverage are retained alongside the ranking summaries.

\subsection{Draw separation and the broader proposal check}

\begin{table}[tb]
\centering\small
\begin{tabular}{lrrrr}
\toprule
Sampler & Original & Shared A & A to B & Chain holdout \\
\midrule
Informed MH & 0.665 & 0.750 & 0.633 & 0.486 \\
Local MH & 0.433 & 0.482 & 0.420 & 0.295 \\
PT & 0.809 & 0.805 & 0.747 & 0.640 \\
\bottomrule
\end{tabular}
\caption{Spearman associations with later drift on all 240 large graphs. Original uses the full early sample in both quantities. Shared A uses half as many early draws in both quantities. A to B diagnoses A but measures drift from disjoint B draws. Chain holdout uses two chains for the diagnostic and two different chains for the outcome. These columns differ in retained information; shared controls and reverse orientations are retained in the companion repository.}
\label{tab:review-split}
\end{table}

For informed, the nonoverlapping-window correlation is 0.633, compared with 0.750 for the same-A control; their paired change is -0.117 with pointwise 95\% interval $[-0.135,-0.099]$. The reverse-window correlation is 0.740. Median within-condition correlations are 0.292 for A-to-B and 0.095 for the chain-held-out check. The reverse chain split gives 0.455.

For local, the nonoverlapping-window correlation is 0.420, compared with 0.482 for the same-A control; their paired change is -0.062 with pointwise 95\% interval $[-0.087,-0.036]$. The reverse-window correlation is 0.479. Median within-condition correlations are 0.435 for A-to-B and 0.155 for the chain-held-out check. The reverse chain split gives 0.332.

For PT, the nonoverlapping-window correlation is 0.747, compared with 0.805 for the same-A control; their paired change is -0.058 with pointwise 95\% interval $[-0.098,-0.020]$. The reverse-window correlation is 0.766. Median within-condition correlations are 0.114 for A-to-B and 0.158 for the chain-held-out check. The reverse chain split gives 0.708.

These controls assess sensitivity to draw reuse and serial dependence; they do not turn drift into posterior error or eliminate graph-level heterogeneity. A change from the original four-chain full-window number also reflects reduced sample size or fewer chains. All new intervals remain exploratory and pointwise.

\begin{table}[tb]
\centering\small
\begin{tabular}{lrrr}
\toprule
Diagnostic & Local MH & PT & Informed MH \\
\midrule
$D_{\mathrm{marg}}$ & 0.420 & 0.747 & 0.633 \\
Score $\widehat R$ & 0.160 & 0.328 & 0.157 \\
Indicator $\widehat R$ & 0.263 & 0.262 & 0.169 \\
Wei\ss/DAR(1) & -0.337 & 0.073 & -0.667 \\
Score $R^*$ & 0.236 & 0.604 & 0.182 \\
Panel $R^*$ & 0.722 & 0.583 & -0.025 \\
Full $R^*$ & 0.666 & 0.694 & 0.078 \\
\bottomrule
\end{tabular}
\caption{All diagnostics recomputed from A and evaluated against B-to-later drift. Each column is a pooled Spearman correlation, with undefined diagnostic values excluded and their counts retained in the companion repository. RF hyperparameters are unchanged; less data can change classifier performance.}
\label{tab:review-comparators}
\end{table}

\paragraph{The categorical reversal is not specific to $n=100$.}
Stratifying the original informed results by size alone, pooling families and
noise levels, gives correlations $0.783$, $-0.754$, and $-0.021$ at $n=20,50,100$,
respectively. Thus the strongly negative pooled correlation cannot be described
as solely an $n=100$ artifact or eliminated simply by conditioning on size.
At $n=50$, three of four family/noise-specific correlations are positive
($0.150$, $0.460$, and $0.018$), while SBM at $q=0.1$ is negative ($-0.426$).
Family/noise aggregation therefore also affects the result. Conversely,
holding each family/noise pair fixed and pooling sizes leaves all four
correlations negative, ranging from $-0.480$ to $-0.879$. These descriptive
partitions show multiple aggregation effects; they do not isolate sparsity,
the dependence correction, or maximum-statistic aggregation as a causal
mechanism. Those explanations remain hypotheses.

\begin{table}[tb]
\centering\small
\begin{tabular}{lrrrr}
\toprule
Variant & Transitions & Work (millions) & $D_{\mathrm{marg}}$ & Window drift \\
\midrule
Broader & 160,000 & 7.680 & 0.983 & 0.258 \\
Broader & 663,333 & 31.840 & 0.979 & 0.306 \\
Single pivot & 160,000 & 31.840 & 0.976 & 0.259 \\
\bottomrule
\end{tabular}
\caption{Matched eight-graph broader-proposal probe at 100 vertices, with four chains per graph. Work counts transposition-delta evaluations per chain. The 160,000-step broader run matches transitions; 663,333 steps match the original 160,000-step single-pivot work within 16 delta evaluations. Comparisons are exploratory, not evidence about every locally balanced proposal.}
\label{tab:review-kernel}
\end{table}

At the matched-work broader budget, none of the eight graphs has both cross-chain disagreement and within-window drift at most $0.1$. The feasible-alignment bound places the score sublevel containing all retained draws below posterior mass $0.01$ on four of the eight graphs. These are complementary checks, not a convergence certificate. The alternative changes candidate construction and move size together; its outcome cannot isolate a single mechanism or establish intrinsic hardness of the target for all informed kernels.

\begin{figure}[tb]
\centering\includegraphics[width=\linewidth]{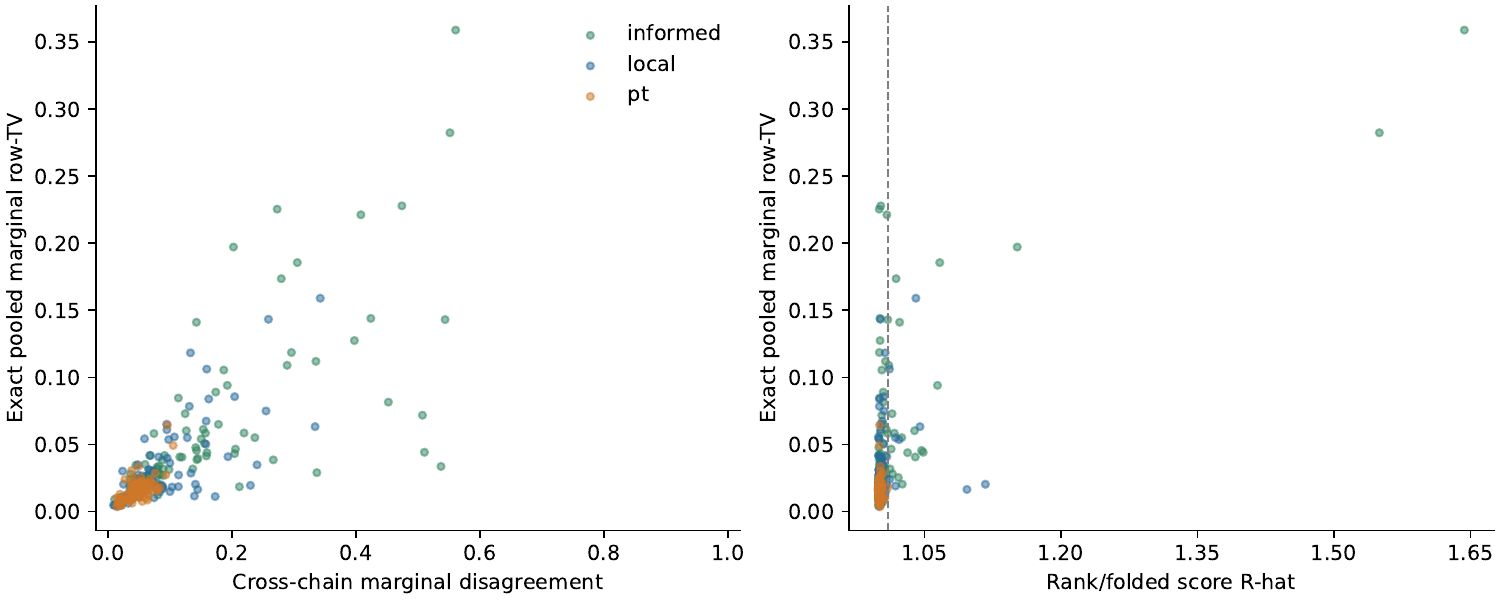}
\caption{New exact audit. Each point is one graph--method combination.
Disagreement and score diagnostics are compared with the pooled posterior
marginal error. Numerical comparisons and sample counts are reported separately
by method; this visualization does not treat the plotted methods as independent
replications of the same graph.}
\end{figure}
\begin{figure}[tb]
\centering\includegraphics[width=\linewidth]{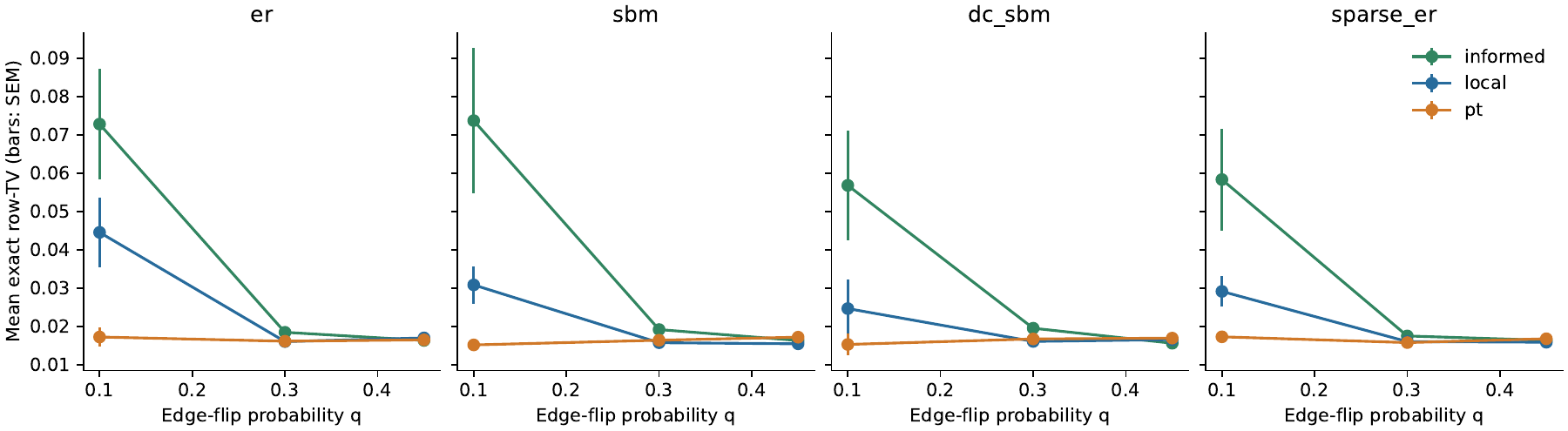}
\caption{Marginal error across graph families and edge-flip regimes. Bars show
standard errors across twenty graphs per condition. Increasing noise changes
the likelihood posterior as well as the input data.}
\end{figure}

\paragraph{Why marginal metrics can look similar.}
Frobenius and Jensen--Shannon comparisons, and selected assignment indicators,
also retain information absent from the score. Strong performance by these
baselines supports the assignment-sensitive interpretation; it would not
establish a unique advantage of the $L^1$ choice. In diffuse targets, empirical
marginal error and disagreement remain positive at finite retained sample
count, even under ideal sampling. Small differences near the IID floor should
not be interpreted as metastable bias.

\begin{figure}[tb]
\centering\includegraphics[width=\linewidth]{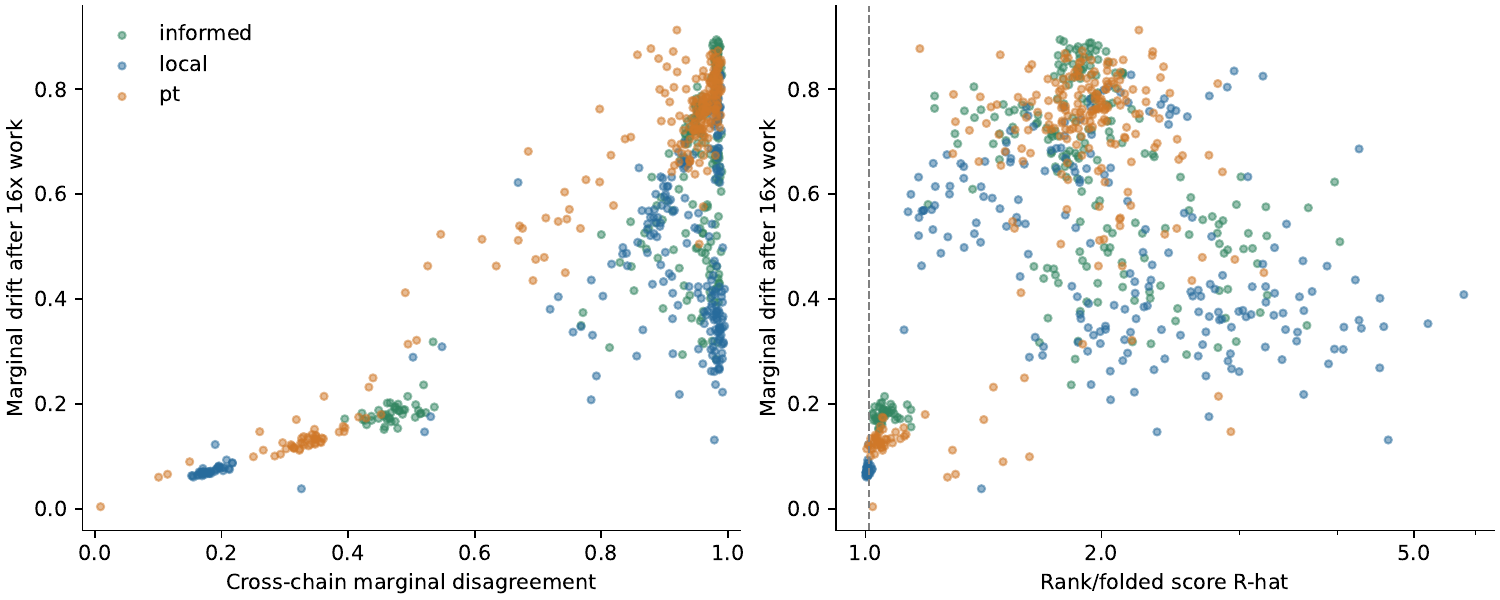}
\caption{Large-graph audit with later marginal drift as the endpoint. Drift is
observable without an exact oracle but is not posterior error. Size and noise
can affect both axes; per-condition associations accompany pooled results.}
\end{figure}

\section{Failure modes and downstream implications}
An exact false-negative example requires no approximate reference. Set $A=B$
to a cycle. The rotation group acts transitively, so every correspondence marginal
is exactly $1/n$ for every temperature. If all reported samples equal the
identity, $\Dm=0$ although pooled row-TV is $1-1/n$. Our actual local MH
stress experiment uses a thirty-cycle, four seeds, $2^{18}$ proposals,
and $\beta\in\{2,4,6,9\}$. It compares a common identity start with four
separated rotation starts. The complete temperature grid is reported; no
single temperature was selected as the only outcome.

At $\beta=6$, common-start row-TV is approximately $0.96665$ while marginal
disagreement is approximately $0.000033$. The score diagnostic is approximately
one. Separated rotation starts instead yield disagreement approximately one,
while their pooled error remains large. At $\beta=9$ the scores are exactly
constant and the modern score diagnostic is undefined. These examples show
why a reporting rule must retain degenerate cases and why adding draws within
the same subset of alignments does not by itself validate posterior uncertainty.

When an exact symmetry is known, averaging a marginal estimate over its
symmetry group cannot increase its $L^1$ error: the averaging operator is a
contraction and leaves the true marginal matrix fixed. This does not apply to
approximate symmetries and does not establish that a chain has explored
inequivalent alignments.

\paragraph{Downstream scope.}
For bounded assignment costs $0\le c_{ia}\le1$, define
$F(P)=n^{-1}\sum_{ia}c_{ia}P_{ia}$. Rowwise TV duality implies
$|F(P)-F(Q)|\le\TV(P,Q)$. Marginal accuracy therefore controls expected
average bounded per-vertex assignment costs. Contact persistence and other
pairwise network observables generally require joint assignment information.
For example, the posterior probability that two mapped vertices form an edge
depends on $\Pr(\pi(i)=a,\pi(j)=b)$, not only on $P_{ia}$ and $P_{jb}$.
The distinction is important when alignment uncertainty is propagated to
network-level scientific quantities.

\section{Applications and future directions}
The experiments above use a deliberately simple posterior, but the auditing
problem is broader: many scientific matching pipelines ultimately consume a
correspondence matrix, a soft alignment, or a small set of plausible mappings.
In such settings, checking only an optimized score or a single representative
matching can miss uncertainty that matters downstream. The most direct future
extensions are therefore domain-specific versions of the same audit: define the
scientific observation model, identify the correspondence summaries that enter
downstream decisions, and validate those summaries rather than only the scalar
objective used to construct them.

\paragraph{Biological networks and connectomics.}
Protein-interaction alignment uses network topology and biological information
to identify functional correspondences~\cite{singh,berg2006,kolar2012}.
Connectome analysis provides an even more literal graph-matching example:
graph matching has been used to pair bilaterally homologous neurons across
hemispheres in nanoscale connectomes~\cite{pedigo2023}. In Bayesian versions of
these tasks, correspondence marginals could support selective annotation
transfer, ambiguity sets, or downstream analyses that retain several plausible
matches rather than committing to one permutation. The present results identify
a computational prerequisite for that workflow: stable objective values should
not be substituted for stable correspondence estimates. Real applications
would additionally require weighted or directed edges, node attributes,
missing interactions, partial matches, and possibly unequal graph sizes.

\paragraph{Spatial and multimodal tissue alignment.}
Alignment is also central to spatial omics. For example, PASTE aligns spatial
transcriptomics slices using an optimal-transport formulation that combines
transcriptional similarity and physical structure, producing soft pairwise
alignments rather than a single hard correspondence~\cite{zeira2022}. These
couplings are not the permutation posterior studied here, and tissue slices can
require many-to-many or unequal-mass matching. Nevertheless, the same audit
question arises whenever soft correspondences are interpreted as uncertainty:
do the reported alignment weights remain stable under independent computation,
and are they accurate on instances where a trusted reference can be constructed?
A natural extension is therefore to adapt assignment-sensitive diagnostics to
partial, attributed, and transport-based alignments and to calibrate them on
synthetic or experimentally registered tissues.

\paragraph{Ecological and spatial population networks.}
Network alignment has been used to compare food webs across ecosystems and to
identify a conserved backbone of species playing similar interaction roles
across communities~\cite{bramon2018}. In that setting, uncertainty over
species correspondences could be propagated to uncertainty over the inferred
backbone rather than reporting only one optimal alignment. Related questions
arise in spatial population models. Gorgi et al.~\cite{moran} derive continuum
descriptions of spatial Moran birth--death and death--birth dynamics, including
heterogeneous environments represented by weighted lattice graphs. Comparing
such environments under uncertain site correspondence would require a weighted,
spatially informed observation model and, for dynamical quantities such as
fixation probabilities or wave speeds, posterior information beyond nodewise
marginals. This makes ecology and spatial population dynamics a useful setting
for extending the audit from correspondence uncertainty to uncertainty in
network-level observables.

\paragraph{Molecular correspondence and chemical reactions.}
Atom mapping in chemical reactions is another direct correspondence problem:
atoms in reactants must be matched to atoms in products, and modern methods can
formulate this as graph matching~\cite{astero2024}. Molecular symmetry is
particularly relevant to the present stress tests because topologically
equivalent atoms can support multiple plausible mappings even when a single
mapping has an excellent structural score. A probabilistic treatment could
therefore benefit from diagnostics that distinguish concentration on one
symmetry-related basin from genuine representation of mapping uncertainty.
The present permutation model is too simple for chemistry---atom and bond types,
reaction constraints, and partial structural changes must enter the target---but
the symmetry-induced diagnostic failure has a close analogue.

\paragraph{Statistical-mechanical and computational perspective.}
The target can be written as a Gibbs distribution with energy $E(\pi)=-S(\pi)$.
Its inverse temperature is fixed by the observation model, rather than a
physical thermometer. Nevertheless, the computational issue has a familiar
analogue in molecular simulation: a low-dimensional observable can appear
stable while the sampled ensemble remains restricted to a subset of states.
Sampling-quality assessments therefore depend on the observables of interest
and on exploration across states~\cite{grossfield}. Replica exchange is used
in physical simulation for the same general exploration problem~\cite{earl}.
The same principle extends to other discrete correspondence problems, including
computer vision and multi-network alignment~\cite{zass2008,lazaro2025}:
convergence should be assessed on the posterior functionals used by the
application, not only on a convenient scalar objective.

\paragraph{Methodological directions.}
Several extensions follow directly from the limitations exposed here. First,
the audit can be generalized from bijections to partial, unequal-size,
weighted, directed, attributed, temporal, and multi-graph alignments. Second,
applications whose downstream quantities depend on edges or motifs require
pairwise or higher-order assignment marginals, motivating diagnostics beyond
nodewise correspondence matrices. Third, exact symmetries suggest quotienting
or symmetry-aware summaries so that exploration among equivalent states is
separated from exploration among scientifically distinct alignments. Fourth,
learned or amortized proposal mechanisms could be combined with the same
small-instance exact calibration used here, so improved speed is not accepted
without verifying uncertainty quality. Finally, domain-specific studies should
propagate alignment uncertainty into the final scientific quantity---for
example functional annotation, neuron homology, tissue registration, ecological
backbones, or reaction centers---and evaluate whether conclusions change when
reference uncertainty is acknowledged.

\section{Conclusion}
Assignment-sensitive diagnostics reveal computational failures that score-only
monitoring can miss in Bayesian graph alignment. The significance is broader
than one matching model: in high-dimensional discrete Bayesian inference,
convergence of a low-dimensional energy or score projection need not imply
convergence of the posterior functionals that actually drive scientific or
algorithmic decisions. The exact benchmark shows a clear benefit over score
$\widehat R$ for the tested informed sampler, while indicator and classifier
diagnostics are competitive and the local-sampler comparison remains uncertain.
At larger sizes, positive associations with later drift persist under
disjoint-window and held-out-chain controls, with attenuation and uneven
within-condition performance. These results support monitoring correspondence
quantities directly, without treating agreement as a convergence certificate.

Reference validation is an equally important part of the audit. Increasing
SMC particle count alone does not resolve the selected disagreements;
reallocating work to rejuvenation improves agreement without establishing
correctness. Feasible-alignment score bounds independently expose severe
sampling failure in concentrated cases. Longer informed runs and one broader
proposal remain unstable, but do not characterize every informed kernel.
Matching score-delta evaluations also does not equalize wall time or all
algorithmic overhead.

The study uses synthetic graphs and finite computational budgets. Exact
calibration is restricted to small instances; larger-graph drift is not
posterior error, and no accurate large-tier reference is established.
Failure-selected follow-ups, shared graph instances, and pointwise intervals
limit generalization and do not support program-wide confirmatory claims.
Within this scope, three features make the study a reusable audit template:
exact small-instance calibration where posterior truth is available, explicit
scrutiny of approximate references rather than assuming that long runs are
correct, and constructive falsification checks when no trustworthy reference
exists. The practical message is to specify the target, validate its
implementation on exact cases, monitor assignment-sensitive summaries alongside
scalar traces, and audit approximate references before using them as ground
truth. Extending this workflow to connectomics, spatial omics, ecological and
molecular networks, and other correspondence problems requires domain-specific
likelihoods and validation of the downstream observables that the alignment is
meant to support.

\end{document}